\documentclass{article}

\usepackage{fullpage}

\usepackage{amssymb}
\usepackage{amsthm}
\usepackage{graphicx}
\usepackage{color}
\usepackage{xspace}
\usepackage{hyperref}
\usepackage[capitalize]{cleveref}
\usepackage{multirow}
\usepackage{mathtools}
\usepackage[table]{xcolor}
\usepackage{wrapfig}
\usepackage{cite}

\usepackage{url}
\urldef{\mailsa}\path|{alfred.hofmann, ursula.barth, ingrid.haas, frank.holzwarth,|
\urldef{\mailsb}\path|anna.kramer, leonie.kunz, christine.reiss, nicole.sator,|
\urldef{\mailsc}\path|erika.siebert-cole, peter.strasser, lncs}@springer.com|

\xspaceaddexceptions{[]\}}
\newcommand{\red}[1]{\textcolor{red}{\textbf{#1}}}

\newcommand{\eg}{e.g.}
\newcommand{\ie}{i.e.}

\newcommand{\pts}{\textit{pts}\xspace}
\newcommand{\ptsp}{\textit{pts$'$}\xspace}
\newcommand{\wts}{\textit{wts}\xspace}
\newcommand{\rts}{\textit{rts}\xspace}

\newcommand{\mts}{\textit{mts}\xspace}

\newcommand{\mrq}{\textit{mRq}\xspace}
\newcommand{\mrp}{\textit{mRp}\xspace}

\newcommand{\M}{\textit{M}\xspace}
\newcommand{\Sh}{\textit{S}\xspace}
\newcommand{\I}{\textit{I}\xspace}

\newcommand{\state}{\textit{state}\xspace}
\newcommand{\busy}{\textit{busy}\xspace}

\newcommand{\data}{\textit{data}\xspace}
\newcommand{\addr}{\textit{addr}\xspace}
\newcommand{\owner}{\textit{owner}\xspace}

\newcommand{\entry}{\textit{entry}\xspace}
\newcommand{\type}{\textit{type}\xspace}

\newcommand{\msg}{\textit{msg}\xspace}
\newcommand{\true}{\textit{True}\xspace}
\newcommand{\false}{\textit{False}\xspace}

\newcommand{\id}{\textit{id}\xspace}

\newcommand{\miss}{\textit{miss}\xspace}

\newcommand{\ch}{\textit{L1}\xspace}
\newcommand{\p}{\textit{L2}\xspace}

\newcommand{\cprq}{\textit{c2pRq}\xspace}
\newcommand{\cprp}{\textit{c2pRp}\xspace}
\newcommand{\pc}{\textit{p2c}\xspace}

\newcommand{\getmsg}{\textit{get\_msg}()\xspace}
\newcommand{\memrq}{\textit{MemRq}\xspace}
\newcommand{\memrp}{\textit{MemRp}\xspace}
\newcommand{\mem}{\textit{Mem}\xspace}

\newcommand{\tom}{\textit{ToM}\xspace}
\newcommand{\tos}{\textit{ToS}\xspace}

\newcommand{\wbrp}{\textit{WBRp}\xspace}

\newcommand{\letbf}{\textbf{let}\xspace}

\newtheorem{definition}{Definition}
\newtheorem{theorem}{Theorem}
\newtheorem{lemma}{Lemma}

\begin{document}

%\mainmatter  % start of an individual contribution

\title{A Proof of Correctness for the Tardis Cache Coherence Protocol}

% \titlerunning{Lecture Notes in Computer Science: Authors' 
%Instructions}

\author{Xiangyao Yu, Muralidaran Vijayaraghavan, Srinivas Devadas
\\ \\
	\{yxy, vmurali, devadas\}@mit.edu \\ \\
	Massachusetts Institute of Technology }
\date{}
%\author{Xiangyao Yu$^\dagger$, Christopher Hughes$^\ddagger$, 
%Nadathur Satish$^\ddagger$, Srinivas Devadas$^\dagger$
% \\ \\
%	$\dagger$ Massachusetts Institute of Technology \\
%	$\ddagger$ Parallel Computing Lab, Intel Labs \\
%}

%
%\authorrunning{Lecture Notes in Computer Science: Authors' 
% Instructions}

%\institute{Springer-Verlag, Computer Science Editorial,\\
%Tiergartenstr. 17, 69121 Heidelberg, Germany\\
%\url{http://www.springer.com/lncs}}


%\toctitle{Lecture Notes in Computer Science}
%\tocauthor{Authors' Instructions}
\maketitle

\begin{abstract}
We prove the correctness of a recently-proposed cache coherence 
protocol, Tardis, which is simple, yet scalable to high processor 
counts, because it only requires $O(\log N)$ storage per cacheline for 
an $N$-processor system. We prove that Tardis follows the sequential 
consistency model and is both deadlock- and livelock-free. Our proof 
is based on simple and intuitive invariants of the system and thus 
applies to any system scale and many variants of Tardis.
\end{abstract}

\setcounter{page}{1}
\section{Introduction} \label{sec:intro}

Tardis~\cite{yu2015} is a recently proposed cache coherence protocol 
that is able to scale to a large number of cores.
Unlike full-map directory 
protocols~\cite{censier1978, tang1976}, Tardis does not keep the 
$O(N)$ ($N$ is the core count) sharer information for each cacheline.  
In Tardis, only the owner ID of each cacheline ($O(\log{N})$) and 
logical timestamps ($O(1)$) for each cacheline are maintained.
Unlike the snoopy bus coherence protocol~\cite{goodman1983}, or limited
directory protocols such as ACKwise \cite{ATAC},
Tardis does not broadcast messages to maintain coherence.

In Tardis, each load or store is assigned a logical timestamp which 
may not agree with the physical time.  The global memory order then 
simply becomes the timestamp order which is explicit in the protocol.  
This makes it much simpler to reason about the correctness of Tardis.  
Despite its simplicity, however, no proof of correctness has yet been 
published.  We provide a simple and straightforward proof in Section \ref{sec:proof};
our proof is simpler than existing proofs for snoopy and directory 
protocols such as \cite{plakal1998, muralidaran2015}.  

%As a result, the protocol is simple to implement and reason about.  
%Further, Tardis only requires $O(\log{N})$ ($N$ is the core count) 
%storage per cacheline as compared to the $O(N)$ storage required by 
%full-map directory protocols. Finally, it does not need to broadcast 
%messages on a bus as do snoopy bus protocols, nor does it broadcast 
%invalidations in the network as do limited directory protocols such 
%as ACKwise \cite{ATAC}.

In this paper, we formally prove the correctness of the Tardis 
protocol by showing that an execution of a program using Tardis 
strictly follows {\it Sequential Consistency} (SC).
%our main result and then generalize to other consistency models. 
We also prove that the Tardis protocol can never deadlock or livelock.

The original Tardis protocol~\cite{yu2015} was designed for a shared 
memory multicore processor. A number of optimization techniques were 
applied for performance improvement. These optimizations, however, may 
not be desirable in other kinds of shared memory systems.
%Moreover, they make it more difficult for theoretical analysis of the 
%protocol.  
Therefore, in this paper we first extract the core algorithm of Tardis 
and prove its correctness.  We then focus on correctness of 
generalizations of the base protocol.

We prove the correctness of Tardis by developing simple and intuitive 
system invariants. Compared to the popular model 
checking~\cite{murphi,tlc2,gigamax} verification techniques, our proof 
technique is able to scale to high
processor counts. More important, the invariants we developed are 
more intuitive to system designers and thus provide more guidance for
system implementation. 

The rest of the paper is organized as follows. The Tardis protocol is 
formally defined in \cref{sec:tardis}. It is proven to obey sequential 
consistency in \cref{sec:proof} and to be deadlock-free and 
livelock-free in \cref{sec:freedom}. In \cref{sec:mem}, we generalize 
the proofs to systems with main memory. \cref{sec:related} describes 
related work and \cref{sec:conclusion} concludes the paper. 

\section{Tardis Coherence Protocol} \label{sec:tardis}

We first present the model of the
shared memory system we use, along with our assumptions, in
\cref{sec:system}.
Then, we introduce system components of the Tardis 
protocol in \cref{sec:format} and formally specify the protocol in 
\cref{sec:protocol}.

\subsection{System Description} \label{sec:system}

%\murali{Also mention that no rule can fire if the corresponding 
%dequeue buffer is empty or enqueue buffer is full}

\cref{fig:system} shows the architecture of a shared memory system 
based on which Tardis will be defined. The processors can execute
instructions in-order or out-of-order but always commit instructions in 
order.  A processor talks to the memory subsystem through a pair of 
{\it local buffers} (LB).  Load and store requests are inserted into 
the {\it memory request buffer} (\mrq) and responses from the memory 
subsystem are in the {\it memory response buffer} (\mrp). 

We model a two-level memory subsystem with all the data fitting into
the L2 caches. The network between L1 and L2 caches is modeled as 
buffers. \cprq contains requests from L1 (child) to L2 (parent), \cprp 
contains responses from L1 to L2, and \pc contains messages (both 
requests and responses) from L2 to L1. For simplicity, all the buffers 
are modeled as FIFOs and \getmsg returns the head message in the 
buffer. However, the protocol also works if the buffers only have the FIFO 
property for the same address. Each L1 cache has a unique \id from 
$1$ to $N$ and each associated buffer has the same \id. An L1 
cacheline or a message in a buffer has the same \id as the L1 cache or 
the buffer it is in.

\begin{figure}[t]
	\centering
	\includegraphics[width=.6\columnwidth]{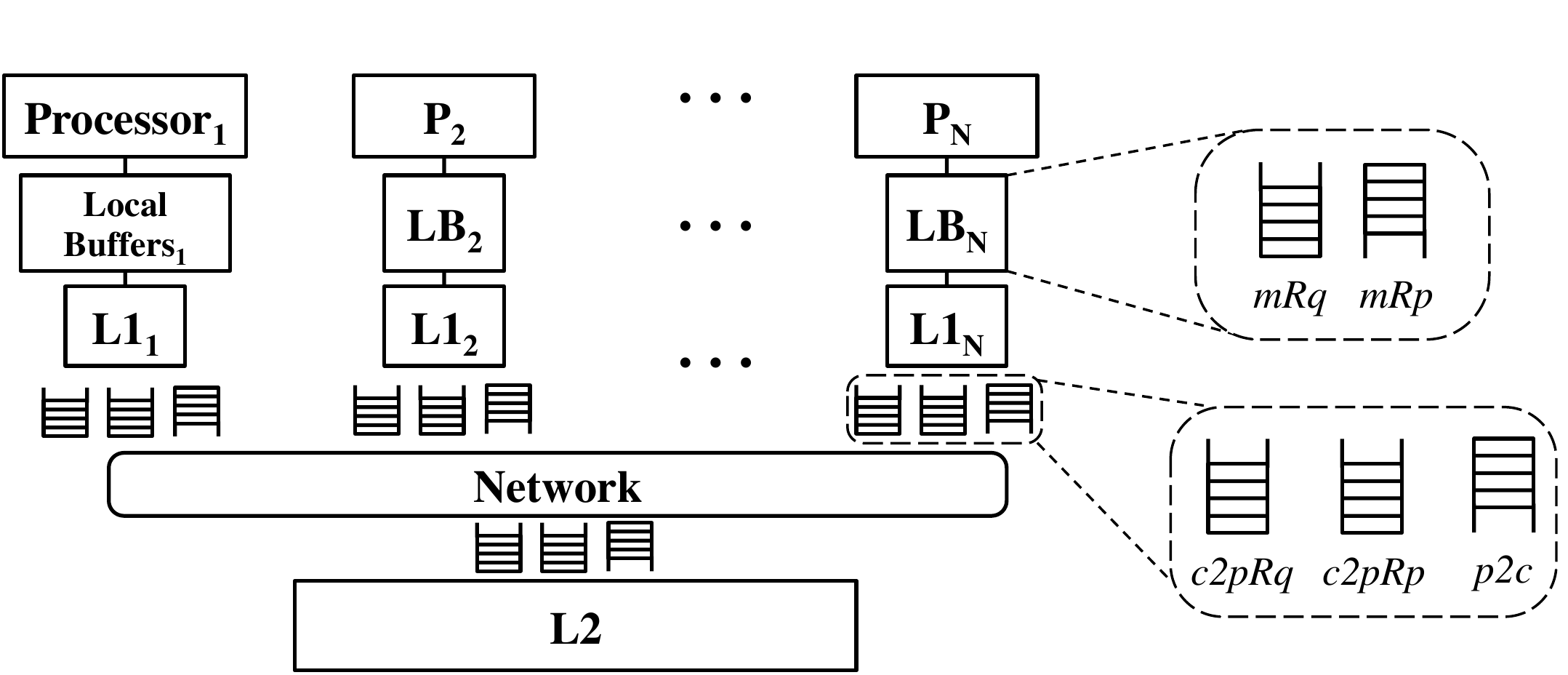}
	\caption{Architecture of a shared memory multicore processor.}  
	\label{fig:system}
\end{figure}

%For ease of discussion, the model defined here is much simplified 
%compared to a practical system. More general models will be discussed 
%in \cref{sec:discussion}.

\subsection{System Components in Tardis} \label{sec:format}

The Tardis protocol is built around the concept of logical timestamps.  
Each memory operation in Tardis has a corresponding timestamp which 
indicates its global memory order. The memory dependency is expressed 
using timestamps and no sharer information is maintained. For 
simplicity, we assume the timestamps to be large enough that they 
never overflow (\eg, 64-bits). Timestamp compression algorithms 
are able to achieve much smaller timestamps (\eg, 16-bits) in 
practice \cite{yu2015}.

At a high level, if a load observes the value of a previous store, 
then the load should be ordered after the store in the logical time 
(and thus global memory order). Similarly, a store should be ordered 
after a load if the load does not observe the stored value. To keep 
track of the timestamps of each operation, every cacheline in Tardis 
has a read timestamp (\rts) and a write timestamp (\wts). The \wts is 
the timestamp of the last store and \rts is the timestamp of the last 
(potential) load to the cacheline (\wts $\leq$ \rts).  Similar to a 
directory protocol, a cacheline can be cached in L1 in either \M or 
\Sh states. Only one L1 can obtain the \M state at any time to modify 
the cacheline, and multiple L1s can share the line in the \Sh state.  
Timestamps are also required for messages in the buffers. 
\cref{tab:format} summarizes the format of caches and buffers modeled 
in the system. The differences between Tardis and a directory protocol 
are highlighted in red.

%Each shared cacheline in the L1 has a lease which is only valid from 
%its \wts to \rts. Trying to access at a timestamp greater than \rts is 
%considered as a miss and a request must be sent to L2 to renew the 
%lease. Different from a directory protocol, when a shared cacheline in 
%L2 receives an exclusive request, no invalidations are sent to the 
%sharers. The cacheline is instantly returned in the \M state.  The 
%protocol can guarantee that the store will happen after all the loads 
%from the sharers in the global memory order.
%
%For ease of discussion, we use the tuple $(s, d, wts, rts)$ to 
%represent a cacheline (in both L1 and L2) where $s$ is the state which 
%can be \M, \Sh or \I, $d$ is the data and $wts$ and $rts$ are the read 
%and write timestamps respectively. Since the L2 is assumed to have 
%infinite size, the state of an L2 cacheline can only be \Sh or hI.  
%
%\red{define the format of network buffers}
%
%\red{add a subsection to discuss the format of the local buffers}
%
%An \mrq entry is $(t, a, pts)$ where the request type $t$ can be {\it 
%Ld, St, M or S} correspinding to store and load requests respectively, 
%$a$ is the address of the requested cacheline and \pts is a timestamp 
%such that the memory operation should have a timestamp greater than or 
%equal to \pts. 
%
%following notation to represent L1 and L2 states.
%
%\red{[optional] have a table to summarize the formats}
%
%\red{TODO TODO. [start with this one first] specify the protocol.}
%
%We speficify the Tardis protocol by listing all the possible state 
%transitions.

\begin{table}%{width=\columnwidth}
	\caption{ System Components}%. \n is the ID of an L1 cache.}
    \begin{center}
	{ \footnotesize
		\begin{tabular}{|c | p{2.8in} | p{1.2in}|}
			\hline
			Component & Format & Message Types
			\\ \hline
			\ch & \ch[\addr] = (\state, \data, \busy, \red{\wts},
			\red{\rts})
			& -
			\\ \hline
			\p & \p[\addr] = (\state, \data, \busy, \owner,
			\red{\wts}, \red{\rts}) & 
			- \\ \hline
			\mrq & \mrq.\entry = (\type, \addr, \data, \red{\pts}) & 
			{\it LdRq, StRq} \\ \hline
			\mrp & \mrp.\entry = (\type, \addr, \data, \red{\pts}) & 
			{\it LdRp, StRp}
			\\ \hline
			\cprq & \cprq.\entry = (\id, \type, \addr, \red{\pts}) & 
			{\it GetS, GetM}
			\\ \hline
			\cprp & \cprp.\entry = (\id, \addr, \data, \red{\wts},
			\red{\rts}) & {\it WBRp}
			\\ \hline
			\pc	& \pc.\entry = (\id, \msg, \addr, \state, \data,
			\red{\wts}, \red{\rts}) & {\it ToS, ToM, WBRq}
			\\ \hline
	
			\end{tabular}
    }
    \end{center}
	\label{tab:format}
	\vspace{-0.25in}
\end{table}

An L1 cacheline contains five fields: \state, \data, \busy, \wts and 
\rts.  The \state can be \M, \Sh or \I. For ease of discussion, we 
define a total ordering among the three states $\I < \Sh < \M$. A 
cacheline has \busy $=$ \true if a request to L2 is outstanding.  
This prevents duplicated requests. An L2 cacheline contains one more 
field \owner, which is the \id of the L1 that exclusively owns the 
cacheline in the \M state. As in L1, \busy in L2 is set when a write back 
request ({\it WBRq}) to an L1 is outstanding.  

%COMMENT you have not defined \pts, and it has different meanings!
Each entry in \mrq contains four fields: \type, \addr, \data and \pts.  
The \type can be \Sh or \M corresponding to a load or store request,
respectively. The \pts is a timestamp specified by the processor and 
the timestamp of the memory operation should be no less than \pts.  
\mrp has the same format as \mrq, but \pts here is the actual 
timestamp of the memory operation.

The format of the messages in the network buffers (\cprq, \cprp and 
\pc) is shown in Table \ref{tab:format}, where the meaning of the 
fields
are as in the cachelines or the messages in \mrq.
All network messages have a field \id which is the \id of the L1 cache that 
the message comes from or goes to.
Messages in \pc have a field \msg, which can be 
either {\it Req} or {\it Resp}; \pc may contain both requests 
and responses from L2 to L1 and \msg differentiates the two types.

\begin{table}%{width=\columnwidth}
	\caption{ State Transition Rules for L1. }
    \begin{center}
	{ \footnotesize
		\begin{tabular}{| p{3.8in} | p{2.4in}|}
		\hline
		{\bf Rules and Condition} & {\bf Action} \\ \hline
		\underline{\bf LoadHit} \newline
			\letbf (\type, \addr, \_, \red{\pts}) = \mrq.\getmsg
			\newline \letbf (\state, \data, \busy, \red{\wts},
			\red{\rts}) = \ch[\addr] \newline
			{\bf condition:} $\neg$ \busy $\wedge$ \type = \Sh
			\red{$\wedge$
			(\state $\ge$ \Sh $\wedge$ \pts $\leq$ \rts)}
		& \mrq.deq() \newline
			\mrp.enq(\type, \addr, \data, \red{max(\pts, \wts)}) 
			\newline
			\textbf{If} (\state = \M) \newline
			\red{\hspace*{.15in}\red{\rts $\coloneqq$ max(\pts, 
		\rts)}}
		\\ \hline
		\underline{\bf StoreHit} \newline
			\letbf (\type, \addr, \data, \red{\pts}) = \mrq.\getmsg
			\newline
			\letbf (\state, \data, \busy, \red{\wts}, \red{\rts}) = 
			\ch[\addr]
			\newline {\bf condition:} $\neg$ \busy $\wedge$ \type = 
		\state = \M
		& \red{\letbf \ptsp = max(\pts, \rts + 1)} \newline
			\mrq.deq() \newline
			\mrp.enq(\type, \addr, \_, \red{\ptsp}) \newline
			\red{\wts $\coloneqq$ \ptsp \newline
			\rts $\coloneqq$ \ptsp}
		\\ \hline
		\underline{\bf L1Miss} \newline
			\letbf (\type, \addr, \data, \red{\pts}) = \mrq.\getmsg 
			\newline
			\letbf (\state, \data, \busy, \red{\wts}, \red{\rts}) = 
			\ch[\addr]
			\newline {\bf condition:} $\neg$ \busy\ $\wedge$ (\state
			$<$ \type \red{$\vee$ \state = \type = \Sh $\wedge$ \pts 
			$>$ \rts})
		& \cprq.enq(\id, \type, \addr, \red{\pts}) \newline
			\busy $\coloneqq$ \true
		\\ \hline
		\underline{\bf L2Resp} \newline
			\letbf (\id, \msg, \addr, \state, \data, \red{\wts},
			\red{\rts}) =
			\pc.\getmsg \newline
			\letbf ({\it l1state}, {\it l1data}, \busy,
			\textit{\red{l1wts}}, \textit{\red{l1rts}}) = \ch[\addr] 
			\newline
			{\bf condition:} \msg = {\it Resp}
		& \pc.deq() \newline
			{\it l1state} $\coloneqq$ \state \newline
			{\it l1data} $\coloneqq$ \data \newline
			\busy $\coloneqq$ \false \newline
			\textit{\red{l1wts}} \red{$\coloneqq$ \wts \newline
			\textit{l1rts} $\coloneqq$ \rts}
		\\ \hline
		%\cellcolor[rgb]{.8,.9,.8}
		\underline{\bf Downgrade} \newline
			\letbf (\state, \data, \busy, \red{\wts},
			\red{\rts}) = \ch[\addr] \newline
			{\bf condition:} $\neg$ \busy $\wedge$ $\exists$
			\state$'. \; \state' < \state \newline
			\hspace*{.6in}\wedge$ LoadHit and StoreHit cannot fire
			%$\wedge$ \mrq.\textit{search}(\addr) $= \epsilon$
		& %\cellcolor[rgb]{.8,.9,.8}
			{\bf If} (\state = \M) \newline
			\hspace*{.15in}\cprp.enq(\id, \addr, \data, \red{\wts},
			\red{\rts}) \newline
			\state $\coloneqq$ \textit{state}$'$
		\\ \hline
		\underline{\bf WriteBackReq} \newline
			\letbf (\state, \data, \busy, \red{\wts}, \red{\rts}) = 
			\ch[\addr] \newline
			{\bf condition:} \pc.\getmsg.\msg = {\it Req} \newline
			\hspace*{.6in}$\wedge$ LoadHit and StoreHit cannot fire
			%$\wedge$ \mrq.\textit{search}(\addr) $= \epsilon$
		& \pc.deq() \newline
			{\bf If} (\state = \M) \newline
				\hspace*{.15in}\cprp.enq(\id, \addr, \data,
				\red{\wts}, \red{\rts}) \newline
				\hspace*{.15in}\state $\coloneqq$ \Sh
		\\ \hline
		\end{tabular}
	}
    \end{center}
	\label{tab:l1-rules}
	\vspace{-0.25in}
\end{table}

\begin{table}%{width=\columnwidth}
	\caption{ State Transition Rules for L2. }
    \begin{center}
	{ \footnotesize
		\begin{tabular}{| p{3.4in} | p{2in}|}
		\hline
		Rules and Condition & Action \\ \hline
		\underline{\bf ShReq\_S} \newline
			\letbf (\id, \type, \addr, \red{\pts}) = \cprq.\getmsg 
			\newline
			\letbf (\state, \data, \busy, \owner, \red{\wts},
			\red{\rts}) =
			\p[\addr] \newline
			{\bf condition:} \type = \Sh $\wedge$ \state = \Sh 
			\newline
			\red{\hspace*{.6in}$\wedge\; \exists\; \ptsp.\; \ptsp \geq 
			\rts
			\wedge \ptsp \geq \pts$}
		& \cprq.deq() \newline \red{\rts $\coloneqq$ \ptsp} \newline
			\pc.enq(\id, {\it Resp}, \addr, \Sh, \data, \red{\wts},
			\red{\ptsp})
		\\ \hline
		%\cellcolor[rgb]{.8,.9,.8}
		\underline{\bf ExReq\_S} \newline
			\letbf (\id, \type, \addr, \red{\pts}) = \cprq.\getmsg 
			\newline
			\letbf (\state, \data, \busy, \owner, \red{\wts},
			\red{\rts}) =
			\p[\addr] \newline
			{\bf condition:} \type = \M $\wedge$ \state = \Sh
		& %\cellcolor[rgb]{.8,.9,.8}
			\cprq.deq() \newline \state $\coloneqq$ \M \newline
			\owner $\coloneqq$ \id \newline
			\pc.enq(\id, {\it Resp}, \addr, \M, \data, \red{\wts},
			\red{\rts})
		\\ \hline
		\underline{\bf Req\_M} \newline
			\letbf (\id, \type, \addr, \red{\pts}) = \cprq.\getmsg 
			\newline
			\letbf (\state, \data, \busy, \owner, \red{\wts},
			\red{\rts}) =
			\p[\addr] \newline
			{\bf condition:} \state = \M $\wedge$ $\neg$ \busy
		& \pc.enq(\owner, {\it Req}, \addr, \_, \_, \_, \_) \newline
			\busy $\coloneqq$ \true
		\\ \hline
		\underline{\bf WriteBackResp} \newline
			\letbf (\id, \addr, \data, \red{\wts}, \red{\rts}) = 
			\cprp.\getmsg \newline
			\letbf (\state, {\it l2data}, \busy, \owner,
			\red{\textit{l2wts}}, \red{\textit{l2rts}}) =
			\p[\addr] \newline
		& \cprp.deq() \newline
			\state $\coloneqq$ \Sh \newline
			{\it l2data} $\coloneqq$ \data \newline
			\busy $\coloneqq$ \false \newline
			\red{\textit{l2wts} $\coloneqq$ \wts \newline
			\textit{l2rts} $\coloneqq$ \rts}
		\\ \hline
		\end{tabular}
	}
    \end{center}
	\label{tab:l2-rules}
	\vspace{-0.25in}
\end{table}

\subsection{Protocol Specification} \label{sec:protocol}

We now formally define the core algorithm of the Tardis protocol.  The 
state transition rules for L1 and L2 caches are summarized in 
\cref{tab:l1-rules} and \cref{tab:l2-rules} respectively, with the 
differences between Tardis and a directory protocol highlighted in 
red. For all rules where a message is enqueued to a buffer, the rule 
can only fire if the buffer is not full.

%\red{mention that L1 and buffer mean the local one.}

%\red{add message names to to the tables?}

%\red{define \getmsg and \msearch}

An important concept in Tardis is the lease. For a cacheline shared in 
an L1 cache, it also contains a lease which expires at the current 
\rts. The data is only valid if the lease has not expired, \ie, \pts 
from the request is less than or equal to \rts. The \rts in the L2 
cache is the maximum of the \rts of all the sharing L1 caches. When a 
shared cacheline in L2 cache gets a {\it GetM} request, it does not 
send invalidations as in a directory protocol, rather, it immediately 
returns exclusive ownership to the requesting processor, which will 
jump ahead in logical time and perform the store at $\rts + 1$. If 
we consider logical time, the store happens after all the loads that
do not observe its value, although in physical time, the store and the
loads may happen simultaneously.  

Specifically, the following six transition rules may fire in an L1 
cache. 

{\bf 1. LoadHit}. LoadHit can fire if the requested cacheline is in 
the \M state or is in the \Sh state and the lease has not expired.  If 
it is in the \M state, then \rts is updated to reflect the latest load 
timestamp.  The \pts returned to the processor is no less than the 
cacheline's \wts, which orders the load after the previous store in 
logical time.

{\bf 2. StoreHit}. StoreHit can only fire if the requested cacheline 
is in the \M state in the L1 cache. Both \wts and \rts are updated to 
the timestamp of the store operation which is at least $\rts + 1$.  
The store is thus logically ordered after all concurrent loads on the 
same address in other L1s.

{\bf 3. L1Miss}. If neither {\it LoadHit} nor {\it StoreHit} can fire 
for a request and the cacheline is not busy, it is an L1 miss and the 
request ({\it GetS} or {\it GetM}) is forwarded to the L2 cache.  
The cacheline is then marked as busy to prevent sending duplicated 
requests.

{\bf 4. L2Resp}. A response from L2 sets all the fields in the L1 
cacheline.  The \busy flag is reset to \false and the cacheline can 
serve the next request in the \mrq.

{\bf 5. Downgrade}. A cacheline in the \M or \Sh states may downgrade 
if the cacheline is not busy and {\it LoadHit} and {\it StoreHit} 
cannot fire. For \M to \Sh or \I downgrade, the cacheline should be 
written back to the L2 in a {\it WBRp} message.  \Sh to \I downgrade, 
however, is silent, \ie, no message is sent. 
%Tardis does not maintain a directory structure and allows silent 
%downgrades.

{\bf 6. WriteBackReq}. When a cacheline in \M state receives a write 
back request, the cacheline is returned to L2 in a {\it WBRp} message 
and the L1 state becomes \Sh.  If the request is to a cacheline in \Sh 
or \I state, the request is simply ignored. This corresponds to the 
case where the line self downgrades after the write back request ({\it WBRq}) is sent from the L2.

The following four rules may fire in the L2 cache. 

{\bf 7. ShReq\_S}. When a cacheline in the \Sh state receives a shared 
request (\ie, {\it GetS}), both the \rts and the returned \pts are set 
to \ptsp which can be any timestamp greater than or equal to the 
current \rts and \pts.  The \ptsp indicates the end of the lease for 
the cacheline. And the cacheline may be loaded at any logical time 
between \wts and \ptsp.  The returned message is a \tos message. 

{\bf 8. ExReq\_S}. When a cacheline in the \Sh state receives an 
exclusive request (\ie, {\it GetM}), the cacheline is instantly 
returned in a \tom message.  Unlike in a directory protocol, {\em no} 
invalidations are sent to the sharers.  The sharing processors may 
still load their local copies of the cacheline and such loads have 
smaller timestamps than the store from the new owner processor. 

{\bf 9. Req\_M}. When a cacheline in the \M state receives a request 
and is not busy, a write back request (\ie, \textit{WBRq}) is sent to 
the current owner.  \busy is set to \true to prevent sending duplicated 
{\it WBRq} requests.

{\bf 10. WriteBackResp}. Upon receiving a write back response (\ie, 
{\it WBRp}), data and timestamps are written to the L2 cacheline. The 
\state becomes \Sh and {\it busy} is reset to \false.

\section{Proof of Correctness} \label{sec:proof}

We now prove the correctness of the Tardis protocol specified in 
\cref{sec:protocol} by proving that it strictly follows sequential 
consistency.
%\footnote{In \cref{app:tso} we discuss
%a proof for the Total Store Order (TSO) consistency model.} 
We first give the definition of sequential consistency in 
\cref{sec:sc} and then introduce the basic lemma
(\cref{sec:basic-lemma}) and timestamp lemmas (\cref{sec:ts-lemma}) 
that are used for the correctness proof.  

Most of the lemmas and theorems in the rest of the paper are proven 
through induction. For each lemma or theorem, we first prove that it is 
true for the initial system state (base case) and then prove that it 
is still true after any possible state transition assuming that it was 
true before the transition. 

In the initial system state, all the L1 cachelines are in the \I state, 
all the L2 cachelines are in the \Sh state and all the network buffers are 
empty. For all cachelines in L1 or L2, $\wts = \rts = 0$ and $\busy = 
\false$. Requests from the processors may exist in the \mrq buffers.  
For ease of discussion, we assume that each initial value in L2 was 
set before the system starts at timestamp $0$ through a store 
operation.

\subsection{Sequential Consistency} \label{sec:sc}

According to Lamport~\cite{lamport1979}, a parallel program is 
sequentially consistent if ``{\it the result of any execution is the 
same as if the operations of all processors were executed in some 
sequential order, and the operations of each individual processor 
appear in this sequence in the order specified by its program}''.  
Using $<_m$ and $<_p$ to represent the global memory order and program 
order per processor respectively, sequential consistency (SC) is 
defined as follows \cite{sorin2011}.

\begin{definition}[Sequential Consistency] An execution of a program 
is sequentially consistent iff 

	{\bf Rule 1:} $\forall X, Y \in \{Ld, St\}$ from the same 
	processor, $X <_p Y \Rightarrow X <_m Y$.

	{\bf Rule 2:} {\it Value of} $L$($a$) = {\it Value of 
Max}$_{<m}\{S(a) | S(a) <_m L(a)\}$, where $L(a)$ and $S(a)$ 
are a load and a store to address $a$, respectively,
and $Max_{<m}$ selects the most recent 
operation in the global memory order.
\end{definition}

In Tardis, the global memory order of sequential consistency is 
expressed using timestamps.  Specifically, Theorem 1 states the 
invariants in Tardis that correspond to the two rules of sequential 
consistency. Here, we use $<_{ts}$ and $<_{pt}$ to represent (logical)
timestamp order that is assigned by Tardis
and physical time order that represents the order of events, respectively.

\begin{theorem}[SC on Tardis] \label{theorem:tardis-sc} An execution 
on Tardis has the following invariants.
% This first invariant is important if we want to prove livelock 
% freedome for the core. But seems like we don't need to prove this.

% {\bf Invariant 1:} \pts of a memory response is no less than the 
% \pts of the corresponding request.

	{\bf Invariant 1:} {\it Value of} $L$($a$) = {\it Value of 
Max}$_{<ts}\{S(a) | S(a) \leq_{ts} L(a)\}$.

	{\bf Invariant 2:} $\forall S_1(a), S_2(a)$, $S_1(a) \neq_{ts} 
	S_2(a)$.

	{\bf Invariant 3:} $\forall S(a), L(a), S(a) =_{ts} L(a)
	\Rightarrow S(a) <_{pt} L(a)$.
\end{theorem}

Theorem~\ref{theorem:tardis-sc} itself is not enough to guarantee 
sequential consistency; we also need the processor model described in 
Definition~\ref{def:inorder-commit}. The processor should commit 
instructions in the program order, which implies physical time order 
and monotonically increasing timestamp order. Both in-order and 
out-of-order processors fit this model. 

\begin{definition}[In-order Commit Processor] 
\label{def:inorder-commit} $\forall X, Y \in \{Ld, St\}$ from the same 
processor, $X <_p Y \Rightarrow X \leq_{ts} Y \wedge X <_{pt} Y$.
\end{definition}

%\murali{Shouldn't it be $X <_{pt} Y$? No two events can happen simultaneously,
%even if they do in terms of real time, their events can be ordered}

Now we prove that given Theorem~\ref{theorem:tardis-sc} and our 
processor model, an execution obeys sequential consistency per 
Definition 1. We first introduce the following definition of the 
global memory order in Tardis.

\begin{definition}[Global Memory Order in Tardis] 
\label{def:mem-order}
\[ X <_m Y \triangleq X <_{ts} Y \vee X =_{ts} Y \wedge X <_{pt} Y. \]
\end{definition}

\begin{theorem}{Tardis with in-order commit processors implements Sequential
Consistency.}
\begin{proof}
According to Definitions 2 and 3,
$X <_p Y \Rightarrow X \leq_{ts} Y \wedge X <_{pt} Y 
\Rightarrow X <_m Y$.  So Rule 1 in Definition 1 is obeyed.

$S(a) \leq_{ts} L(a) \Rightarrow S(a) <_{ts} L(a) \vee S(a) =_{ts} L(a)$.  By
Invariant 3 in Theorem 1, this implies $S(a) \leq_{ts} L(a) \Rightarrow S(a)
<_{ts} L(a) \vee S(a) =_{ts} L(a) \wedge S(a) <_{pt} L(a) $. Thus, from
Definition \ref{def:mem-order}, $S(a) \leq_{ts} L(a) \Rightarrow S(a) <_{m}
L(a)$. We also have $S(a) <_m L(a) \Rightarrow S(a) \leq_{ts} L(a)$ from
Definition \ref{def:mem-order}.
%According to Invariant 3 in Theorem 1, $L(a) =_{ts} S(a) \Rightarrow 
%S(a) <_{pt} L(a)$. Then, according to our translation, $S(a) \leq_{ts} 
%L(a) \Rightarrow S(a) <_{ts} L(a) \vee S(a) =_{ts} L(a) \wedge S(a) 
%<_{pt} L(a) \Rightarrow S(a) <_m L(a)$, and we also have $S(a) <_m 
%L(a) \Rightarrow S(a) \leq_{ts} L(a)$.
So $\{S(a) | S(a) \leq_{ts} 
L(a)\} = \{S(a) | S(a) <_m L(a)\}$.  According to Invariant 2 in Theorem 1, 
all the elements in $\{S(a) | S(a) <_m L(a)\}$ have different 
timestamps, which means $<_m$ and $<_{ts}$ indicate the same ordering.  
Finally, Invariant 1 in Theorem 1 becomes Rule 2 in Definition 1.
\end{proof}
\end{theorem}

In the following two sections, we focus on the proof of Theorem 1.  
%Our proof has a layered structure. We first prove a basic lemma in 
%\cref{sec:basic-lemma} which lays the foundation for timestamp lemmas 
%introduced in \cref{sec:ts-lemma}. Finally, we prove Theorem 1 using 
%these timestamp lemmas.

\subsection{Basic Lemma} \label{sec:basic-lemma}

We first give the definition of a clean block for ease of discussion.

\begin{definition}[Clean Block] A clean block can be an L2 cacheline 
in \Sh state, or an L1 cacheline in \M state, or a {\it ToM} or {\it 
WBRp} message in a network buffer.
\end{definition}

\begin{wrapfigure}{r}{.4\columnwidth}
	\vspace{-.5in}
	\centering
	\includegraphics[width=.4\columnwidth]{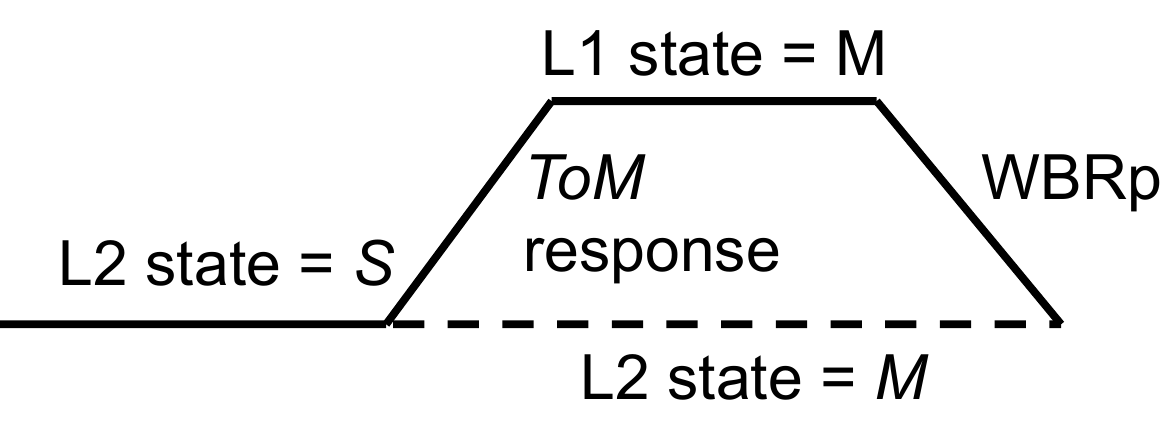}
	\vspace{-.3in}
	\caption{A visualization of Lemma~\ref{lemma:basic}}
	\label{fig:clean}
\end{wrapfigure}

\begin{lemma} \label{lemma:basic} $\forall$ address $a$, at most one 
clean block exists.
\end{lemma}

The basic lemma is an invariant about the cacheline states and the 
messages in network buffers. No timestamps are involved.

A visualization of Lemma~\ref{lemma:basic} is shown in 
\cref{fig:clean} where a solid line represents a clean block. When the L2 
state for an address is \Sh, no L1 can have that address in \M state, and no {\it 
ToM} and {\it WBRp} may exist. Otherwise if the L2 state is \M, either a 
{\it ToM} response exists, or an L1 has the address in \M state, or a 
{\it WBRp} exists. 
%cacheline is clean.  Specifically, exactly one of the following 
%statements is always true.  \begin{itemize}
%\setlength\itemsep{0em}
%COMMENT better to just say \I or \S here I think; I also changed p2c 
%to \pc
%\item State of $a$ is less than \M in the L2.
%\item One {\it ToM} response for $a$ exists in one \pc buffer.
%\item State of $a$ is \M in one L1.
%\item One {\it WBRp} for $a$ exists in one \cprp buffer.
%\end{itemize}
%COMMENT don't you need to show a base case?  That is, Lemma 1 is true 
%at initialization?
Intuitively, Lemma~\ref{lemma:basic} says only one block in the system 
can represent the latest data value.

\begin{proof}[Lemma~\ref{lemma:basic} Proof]
For the base case, the lemma is trivially true since exactly one clean 
block exists for each address and the block is an L2 cacheline in 
\Sh state. We now consider all the possible transition rules that may 
create a new clean block. 

Only the {\it ExReq\_S} rule can create a \tom response. However, the 
rule changes the state of the L2 cacheline from \Sh to \M and thus 
removes a clean block. Only the {\it L2Resp} rule can change an L1 
cacheline state to \M.  However, it removes a \tom response from the 
\pc buffer. Both {\it Downgrade} and {\it WriteBackReq} can enqueue 
\wbrp messages and both will change the L1 cacheline state from \M to 
\Sh or \I. Only {\it WriteBackResp} changes the L2 cacheline to \Sh 
state but it also dequeues a \wbrp from the buffer.

In all of these transitions, a clean block is created and another one 
is removed. By the induction hypothesis, at most one clean block per 
address exists before the current transition, and still at most one 
clean block for the address exists after the transition. For other 
transitions not listed above, no clean block can be created so at most 
one clean block per address exists after any transition, proving the 
lemma. 
\end{proof}

\subsection{Timestamp Lemmas} \label{sec:ts-lemma}

\begin{lemma} \label{lemma:clean-latest}
At the current physical time, a clean block has the following 
invariants.

{\bf Invariant 1} Its \rts is no less than the \rts of all the other 
blocks (in caches and messages) with respect to the same address.

{\bf Invariant 2} Till the current physical time, no store has 
happened to the address at timestamp ts such that ts $>$ \rts.
\end{lemma}

\begin{proof}%[Lemma~\ref{lemma:clean-latest} Proof]
We prove the lemma by induction on the transition sequence. 
For the base case, only one block exists per address so Invariant 1 is 
true. All the stores so far happened at timestamp 0 which equals the 
\rts of all the clean blocks, proving Invariant 2.

According to Lemma 1, for an address, exactly one clean block 
exists. By the induction hypothesis, if no timestamp changes and no clean 
block is generated, Invariant 1 is still true after the 
transition. By the transition rules, \wts or \rts can only be increased if
the block is an L2 cacheline in the \Sh state or an L1 cacheline in the \M 
state. In both cases the block is clean. After the transition, the 
\rts of the clean block increases and is still no less than the \rts 
of other cachelines with the same address. 

Similarly, by the induction hypothesis, Invariant 2 is true after the 
transition if no store happens and no clean block is generated.  Only 
{\it StoreHit} can perform a store to a clean block, which changes 
both \wts and \rts to be $\max(\pts, \rts + 1)$. For the stored 
cacheline, since no store has happened with timestamp $ts$ such that 
${\it ts} > {\it old\_}\rts$ (induction hypothesis), after the 
transition, no store, including the current one, has happened with 
timestamp \textit{ts} such that ${\it ts} > \max(\pts, {\it old\_}\rts 
+ 1) > {\it old\_}\rts$.

Consider the last case where a clean block is generated at the current 
transition. Here, according to Lemma~\ref{lemma:basic}, another clean 
block disappears. In all the transitions, the \rts of the new 
clean block equals the \rts of the old block. Thus, both 
invariants are still true after the transition.
\end{proof}

\begin{lemma} \label{lemma:ts-store} For any block B 
%(\ie, not a cacheline in \I state) 
in a cache or a message ({\it WBRp}, {\it ToS} and {\it ToM}), the 
data value associated with the block comes from a store St which has 
happened before the current physical time, and no other store $St'$ 
has happened such that $St.ts < St'.ts \leq B.\rts$, where St.ts is 
the timestamp of the store St and B.\rts is the \rts of block B.
\end{lemma}

\begin{proof}%[Lemma \ref{lemma:ts-store} proof]
We prove the lemma by induction on the transition sequence. For the 
base case, each block has a corresponding store which happened before
the system started and is thus before the current physical time. It is 
also the only store happened per address.  Therefore the hypothesis is 
true.

We first prove that after a transition, for each block, there exists a 
store $St$ which creates the data of the block and $St$ happened 
before the current physical time. Consider the case where the data of 
a block does not change or comes from another block. Then, by the induction 
hypothesis, $St$ must exist for the block after the transition. The 
only transition that changes the data of a block is {\it StoreHit}.  
After the store, however, the data stored in the cacheline comes from the 
current store which has just happened in physical time. 

We now prove the second part of the lemma, that for any block $B$, no 
store $St' \neq St$ has happened such that $St.ts < St'.ts \leq 
B.\rts$. By the induction hypothesis, for the current transition, if 
no \data or \rts is changed in any block or if a block copies \data 
and \rts from an existing block, then the hypothesis is still true 
after the transition. The only cases in which the hypothesis may be 
violated are when the current transition changes \rts or \data for some 
block, which is only possible for {\it LoadHit}, {\it StoreHit} and 
{\it ShReq\_S}.

For {\it LoadHit}, if the cacheline is in the \Sh state, then \rts remains
unchanged. Otherwise, the cacheline must be a clean block, in which case
\rts is increased. Similarly, {\it ShReq\_S} increases the \rts and
the cacheline must be a clean block again.
By Invariant 2 in  
Lemma~\ref{lemma:clean-latest}, no store has happened to the address 
with timestamp greater than \rts. And thus after the \rts is 
%CHECK changed this statement
increased, no store can have happened with timestamp between the
old \rts to the new \rts. By 
the induction hypothesis, we also have that no store $St'$ could have 
happened such that $St.ts < St'.ts \leq old\_rts$. These two inequalities 
together prove the hypothesis.

%CHECK add max
For {\it StoreHit}, both \rts and \data are modified. For the stored 
cacheline, after the transition, $St.ts$ = \wts = \rts = {\it 
max(\pts, old\_}\rts + 1).  Thus, no $St'$ may exist in this situation.
For all the other cachelines, by Invariant 1 in 
Lemma~\ref{lemma:clean-latest}, their \rts is no greater than the old 
\rts of the stored cacheline and is thus smaller than the timestamp of 
the current store. By the induction hypothesis, no store $St'$ exists for those 
blocks before the transition. % which disproves the theorem.  
Thus, in the overall system, no such store $St'$ can exist, proving 
the lemma.
\end{proof}

%\begin{lemma} \label{lemma:no-store-in-between} For all cachelines in 
%caches or messages ({\it WBRp}, {\it ToS} and {\it ToM}), no store has 
%happened at timestamp {\it ts} to the address such that \wts $<$ {\it 
%ts} $\leq$ \rts.
%\end{lemma}
%
%%COMMENT I am not clear on what you are assuming here: Lemma 3 is true before the transition? That seems to be the argument in the first paragraph
%\begin{proof}[Lemma \ref{lemma:no-store-in-between} Proof]
%For a transition, if no store happens, no timestamps change and no 
%cachelines are created, Lemma 3 is true after the transition.  
%
%If a {\it StoreHit} happens to a clean cacheline, the timestamp for the 
%store {\it ts = \rts + 1}. According to Lemma 2, all the other 
%cachelines with the same address have \rts no greater than the clean 
%cacheline's \rts and are thus less than {\it ts}. So after the store, 
%Lemma 3 holds for all the cachelines.
%
%If a {\it ShReq\_S} happens to an L2 cacheline in \Sh state, the \rts 
%is increased. According to the Invariant 2 in Lemma 2, no store has 
%timestamp greater than the \rts before the transition. And thus no 
%store has timestamp between the old \rts and the new \rts, and Lemma 3 
%is true. 
%
%Finally, if a cacheline is created, its \wts and \rts are copied from 
%another cacheline which follows Lemma 3 before the transition, as a 
%result, the new cacheline also follows the invariant after the 
%transition.  
%\end{proof}

Finally, we prove Theorem 1. 

\begin{proof}[Theorem 1 Proof]
According to Lemma \ref{lemma:ts-store}, for each $L(a)$, the loaded 
data is provided by an $S(a)$ and no other store $S'(a)$ has happened 
between the timestamp of $S(a)$ and the \rts. And thus no $S'(a)$ has 
happened between the timestamp of $S(a)$ and the timestamp of the load 
which is no greater than \rts by the transition rules. Therefore,
Invariant 1 in Theorem~\ref{theorem:tardis-sc} is true. 

By the transition rules, a new store can only happen to a clean block and 
the timestamp of the store is $\max(\pts, \rts + 1)$. By Invariant 2 in 
Lemma~\ref{lemma:clean-latest}, for a clean block at the current 
physical time, no store to the same address has happened with 
%CHECK added old here in a couple of places
timestamp greater than the old \rts of the cacheline. Therefore, for each 
new store, no store to the same address so far has the same timestamp 
as the new store, because the new store's timestamp is strictly greater
than the old \rts.  And thus no two stores to the same address may have 
the same timestamp, proving Invariant 2.

Finally, we prove Invariant 3. If $S(a) =_{ts} L(a)$, by Invariant 1 
in Theorem~\ref{theorem:tardis-sc}, $L(a)$ returns the data stored by 
$S(a)$. Then by Lemma~\ref{lemma:ts-store}, the store $S(a)$ must have 
happened before $L(a)$ in the physical time.
\end{proof}

\section{Deadlock and Livelock Freedom} \label{sec:freedom}

In this section, we prove that the Tardis protocol specified in 
\cref{sec:tardis} is both deadlock-free (\cref{sec:deadlock}) and 
livelock-free (\cref{sec:livelock}). 

\subsection{Deadlock Freedom} \label{sec:deadlock}

%We have the following theorem for deadlock freedom. 

\begin{theorem}[Deadlock Freedom] \label{theorem:dl} After any 
sequence of transitions, if there is a pending request from any 
processor, then at least one transition rule (other than the Downgrade 
rule) can fire.
\end{theorem}

Before proving the theorem, we first introduce and prove several 
lemmas.  

\begin{lemma} \label{lemma:l1-busy}
If an L1 cacheline is busy, either a {\it GetS} or {\it GetM} request 
exists in its \cprq buffer or a {\it ToS} or {\it ToM} response exists 
in its \pc buffer.
\end{lemma}

\begin{proof}%[Lemma~\ref{lemma:l1-busy} Proof]
This can be proven through induction on the transition 
sequence. In the base case, all the L1 cachelines are non-busy and the 
hypothesis is true. An L1 cacheline can only become busy through the 
{\it L1Miss} rule, which enqueues a request to its \cprq buffer. A 
request can only be dequeued from \cprq through the {\it ShReq\_S} or 
{\it ExReq\_S} rule, which enqueues a response into the same L1's \pc 
buffer. Finally, whenever a message is dequeued from the \pc buffer 
({\it  L2Resp} rule), the L1 cacheline becomes non-busy,
proving the lemma.
\end{proof}

\begin{lemma} \label{lemma:l2-busy-m}
If an L2 cacheline is busy, the cacheline must be in state \M.
\end{lemma}

\begin{proof}% [Lemma~\ref{lemma:l2-busy-m} Proof]
This lemma can be proven by induction on the transition 
sequence. For the base case, no cachelines are busy and the hypothesis 
is true. Only {\it Req\_M} makes an L2 cacheline busy but the 
cacheline must be in the \M state. Only {\it WriteBackResp} downgrades an 
L2 cacheline from the \M state but it also makes the cacheline non-busy.
\end{proof}

\begin{lemma} \label{lemma:l2-m}
For an L2 cacheline in the \M state, the \id of the clean block for the 
address equals the owner of the L2 cacheline.
\end{lemma}

\begin{proof}%[Lemma~\ref{lemma:l2-m} Proof]
According to Lemma~\ref{lemma:basic}, exactly one clean block exists
for the address. If the L2 state is \M, the clean block can be a {\it 
ToM} response, an L1 cacheline in the \M state or a {\it WBRp}. We prove 
the lemma by induction on the transition sequence.  

The base case is true since no L2 cachelines are in the \M state.  We 
only need to consider cases wherein a clean block is created. When 
{\it ToM} is created ({\it ExReq\_S} rule), its \id equals the owner 
in the L2 cacheline. When an L1 cacheline in the \M state is created 
({\it L2Resp} rule), its \id equals the \id of the {\it ToM} response.  
When a {\it WBRp} is created ({\it WriteBackReq} or {\it Downgrade} 
rule), its \id equals the \id of the L1 cacheline.  By the induction 
hypothesis, the \id of a newly created clean block always equals the 
\owner in the L2 cacheline which does not change as long as the L2 
cacheline is in the \M state.
\end{proof}

\begin{lemma} \label{lemma:l2-busy}
For a busy cacheline in L2, either a {\it WBRq} or a {\it WBRp} exists
for the address with \id matching the owner of the L2 cacheline.
\end{lemma}

\begin{proof}%[Lemma~\ref{lemma:l2-busy} Proof]
We prove the lemma by induction on the transition sequence. For the 
base case, no cacheline is busy and thus the hypothesis is true. We 
only need to consider the cases where an L2 cacheline is busy after 
the current transition, \ie, $\neg\busy \Rightarrow \busy$ and $\busy 
\Rightarrow \busy$.

Only the {\it Req\_M} rule can cause a $\neg\busy \Rightarrow \busy$ 
transition and the rule enqueues a {\it WBRq} into \pc with \id 
matching the \owner and therefore the hypothesis is true. 

For $\busy \Rightarrow \busy$, the lemma can only be violated if a 
{\it WBRq} or {\it WBRp} with matching \id is dequeued. However, when 
a {\it WBRp} is dequeued, the cacheline becomes non-busy in L2 
(\textit{WriteBackResp} rule). If a {\it WBRq} is dequeued and the L1 
cacheline is in the \M state, a {\it WBRp} is created with a matching \id.  
So the only case to consider is when the {\it WBRq} with matching \id 
is dequeued, and the L1 cacheline is in the \Sh or \I states, and no other 
{\it WBRq} exists in the same \pc buffer and no {\it WBRp} exists in 
the \cprp buffer. 

The L2 cacheline can only become \busy by sending a {\it WBRq}.  
The fact that the dequeued {\it WBRq} is the only {\it WBRq} in the 
\cprq means that the L2 cacheline has been busy since the dequeued
{\it WBRq} was sent (otherwise another {\it WBRq} will be sent when 
the L2 cacheline becomes busy again). Since \pc is a FIFO, when the 
{\it WBRq} is dequeued, the messages in the \pc must be sent after the 
{\it WBRq} was sent. By transition rules, the L2 cacheline cannot 
send {\it ToM} while being busy, so no {\it ToM} may exist in the \pc 
buffer when {\it WBRq} dequeues.  
% Let's say a {\it ToM} response exists to the same L1 in \pc buffer.  
% Since \pc is a FIFO buffer, {\it ToM} would have to be sent after 
% {\it WBRq} has been sent since {\it WBRq} is in the head of \pc for 
% that cache. And by transition rules, L2 cannot send {\it ToM} 
% responses while being busy.So there does not exist a  {\it ToM} 
% response in the \pc when the {\it WBRq} dequeues.
As a result, no clean block exists with \id = \owner. Then, by 
Lemma~\ref{lemma:l2-m}, no clean block exists for the address (L2 is 
in the \M state
because of Lemma~\ref{lemma:l2-busy-m}) which 
contradicts Lemma~\ref{lemma:basic}.
\end{proof}

\begin{proof}[Theorem \ref{theorem:dl} Proof]
If any message exists in the \cprp buffer, the {\it WriteBackResp} 
rule can fire. Consider the case where no message exists in \cprp 
buffer.  If any message exists in the \pc buffer's head, the {\it 
L2Resp} rule can fire, or the {\it WriteBackReq}, % rule 
%({\it WriteBackReq} rule can fire because \cprp buffer is empty) 
{\it LoadHit} or {\it StoreHit} rule can fire.  For the theorem to be 
violated, no messages can exist in the \cprp or \pc buffer.  Then, 
according to Lemma \ref{lemma:l2-busy}, all cachelines in L2 are 
non-busy.

Now consider the case when no message exists in \cprp buffer or \pc 
buffer and
a {\it GetS} or {\it GetM} request exists 
in \cprq for an L1 cache. Since the L2 is not busy, one of {\it 
ShReq\_S}, {\it ExReq\_S} and {\it Req\_M} can fire, which enqueues a 
message into the \pc buffer. 

Consider the last case where there is no message in any network 
buffer. By Lemma~\ref{lemma:l1-busy}, all L1 cachelines are non-busy. 
By the hypothesis, there must be a request in \mrq for some processor. Now 
if the request is a hit, the corresponding hit rule (\textit{LoadHit} 
or \textit{StoreHit}) can fire. Otherwise, the \textit{L1Miss} rule 
can fire, sending a message to \cprq.
\end{proof}

\subsection{Livelock Freedom} \label{sec:livelock}

Even though the Tardis protocol correctly follows sequential 
consistency and is deadlock-free, livelock may still occur if the 
protocol is not well designed. For example, for an L1 miss, the {\it 
Downgrade} rule may fire immediately after the {\it L2Resp} but before 
any {\it LoadHit} or {\it StoreHit} rule fires. As a result, the {\it 
L1Miss} needs to be fired again but the {\it Downgrade} always happens 
after the response comes back, leading to livelock. We avoid this 
possibility by only allowing {\it Downgrade} to fire when neither {\it 
LoadHit} nor {\it StoreHit} can fire. 

To rigorously prove livelock freedom, we need to guarantee that some  
transition rule should eventually make forward progress and no 
transition rule can make backward progress. We give the following 
definition of livelock freedom. 

\begin{theorem} \label{theorem:livelock}
After any sequence of transitions, if there exists a pending request 
from any processor, then within a finite number of transitions, some
request at some processor will dequeue.
\end{theorem}

In order to prove the theorem, we will show that for every transition 
rule, at least one request will make forward progress and move one 
step towards the end of the request and at the same time no other 
request makes backward progress; or if no request makes forward or 
backward progress for the transition, we show that such transitions 
can only be fired a finite number of times. Specifically, we 
define forward progress as a lattice of system states where each 
request in \mrq (load or store) has its own lattice.  
\cref{tab:lattice} shows the lattice for a request where the lower 
parts in the lattice correspond to the states with more forward 
progress.
%and lattices for all the requests comprise the whole system lattice.  
We will prove livelock freedom by showing that for any state 
transition in any cache, any request either moves down the lattice 
(making forward progress) or stays at the current position but never 
moves up.  Moreover, transitions which keep the state of every request 
staying at the same position in the lattice can only occur a finite 
number of times. Specifically, we will prove the following lemma.

\begin{table}[htbp]%{width=\columnwidth}
	\caption{Lattice for a request. For a load request, \ch.\miss
	$\triangleq$ (\ch.\state = \I $\vee$ \ch.\state = \Sh
	$\wedge$ \pts $>$ \ch.\rts).
	For a store request, \ch.\miss $\triangleq$ (\ch.\state  $<$ \M).
	\textit{bufferName\_exist} means a message exists in the buffer
	and \textit{bufferName\_rdy} means that the message is the head of
	the buffer. \textit{bufferName\_rdy} implies
	\textit{bufferName\_exist}.} \begin{center}
	{ \footnotesize
		\begin{tabular}{ c | l }
		1 & {L1.miss $\wedge$ $\neg$L1.busy}
		%$|$ {\it Req} $\coloneqq$ ((\Sh, \addr, \_, \pts) = 
		%\mrq.\getmsg()), \miss $\coloneqq$ \state = \I $\vee$ \state 
		%= \Sh $\wedge$ \pts $>$ $\rts$
		\\
		2 & {L1.miss $\wedge$ L1.busy $\wedge$ c2pRq\_exist $\wedge$ 
		$\neg$ c2pRq\_rdy}
		% $\wedge$ $\neg$p2cRp\_exist}
		\\
		3 & {L1.miss $\wedge$ L1.busy $\wedge$ c2pRq\_rdy $\wedge$ 
		L2.state = \M
		$\wedge$ $\neg$L2.busy}
		\\
		4 & {L1.miss $\wedge$ L1.busy $\wedge$ 
		c2pRq\_rdy $\wedge$ L2.state = \M
		$\wedge$ L2.busy $\wedge$
		p2cRq\_exist $\wedge$ $\neg$p2cRq\_rdy}
		\\
		5 & {L1.miss $\wedge$ L1.busy $\wedge$ 
		c2pRq\_rdy $\wedge$ L2.state = \M
		$\wedge$ L2.busy $\wedge$ p2cRq\_rdy $\wedge$ ownerL1.state = 
		\M}
		\\
		6 & {L1.miss $\wedge$ L1.busy $\wedge$ 
		c2pRq\_rdy $\wedge$ L2.state = \M
		$\wedge$ L2.busy $\wedge$ p2cRq\_rdy $\wedge$ ownerL1.state 
		$<$ \M}
		\\
		7 & {L1.miss $\wedge$ L1.busy $\wedge$ c2pRq\_rdy $\wedge$
		L2.state = \M $\wedge$ L2.busy $\wedge$ $\neg$p2cRq\_exist} \\
		%7 & {L1.miss $\wedge$ L1.busy $\wedge$ c2pRq\_rdy $\wedge$ 
		%L2.state = \M $\wedge$ L2.busy $\wedge$ c2pRp\_exist $\wedge$
		%$\neg$c2pRp\_rdy}
		%\\
		%8 & {L1.miss $\wedge$ L1.busy $\wedge$ c2pRq\_rdy $\wedge$ 
		%L2.state = \M
		%$\wedge$ L2.busy $\wedge$ c2pRp\_rdy}
		%\\
		8 & {L1.miss $\wedge$ L1.busy $\wedge$
		c2pRq\_rdy $\wedge$ L2.state = \Sh} %$\wedge$ \pts $>$ L2.rts}
		\\
		%10 & {\it Req $\wedge$ L1.miss $\wedge$ L1.busy $\wedge$ 
		%c2pRq\_rdy $\wedge$ L2.state = \Sh $\wedge$ \pts $\leq$ 
		%L2.rts}
		%\\
		9 & {L1.miss $\wedge$ L1.busy $\wedge$ p2cRp\_exist $\wedge$ 
		$\neg$p2cRp\_rdy}
		% $\wedge$ $\neg$c2pRq\_exist}
		\\
		10 & {L1.miss $\wedge$ L1.busy $\wedge$ p2cRp\_rdy}
		\\
		11 & {$\neg$L1.miss}
		\\
		%13 & {\it $\neg$Req}
		
		\end{tabular}
	}
    \end{center}
	\label{tab:lattice}
	\vspace{-0.25in}
\end{table}

\begin{lemma} \label{lemma:lattice} For a state transition except 
Downgrade, WriteBackReq and WriteBackResp, either a request dequeues 
from the \mrq or at least one request will move down its lattice. For 
all the state transitions, no request will move up its lattice.  
Further, the system can only fire Downgrade, WriteBackReq and 
WriteBackResp for a finite number of times without firing other 
transitions.
\end{lemma}

We need the following lemmas before proving Lemma~\ref{lemma:lattice}.

\begin{lemma} \label{lemma:req-mutex}
If an L1 cacheline is busy, then exactly one request ({\it GetS} or 
{\it GetM} in \cprq) or response ({\it ToS} or {\it ToM} in \pc) 
exists for the address and the L1 cache. If the L1 cacheline is 
non-busy, then no request or response can exist in its \cprq and \pc. 

\begin{proof}%[Lemma~\ref{lemma:req-mutex} Proof]
This lemma is a stronger lemma than Lemma~\ref{lemma:l1-busy}.
We prove this by the induction on the transition sequence. For the base 
case, all the L1 cachelines are non-busy and no message exists and 
thus the lemma is true. 

We only need to consider the cases where the busy flag changes or any 
request or response is enqueued or dequeued. Only the {\it L1Miss}, 
{\it L2Resp}, {\it ShReq\_S} and {\it ExReq\_S} rules need to be 
considered. 

For {\it L1Miss}, a request is enqueued to \cprq and the L1 cacheline 
becomes busy. For {\it L2Resp}, a response is dequeued and the L1 
cacheline becomes non-busy. For {\it ShReq\_S} and {\it ExReq\_S}, a 
request is dequeued but a response in enqueued. By the induction 
hypothesis, after the current transition, the hypothesis is still true 
for all the cases above, proving the lemma.
\end{proof}
\end{lemma}

\begin{lemma} \label{lemma:l1-busy-mrq}
If an L1 cacheline is busy, there must exist a request at the head of 
the \mrq buffer for the address and the request misses in the L1.  
\end{lemma}

\begin{proof}
For the base case, all L1 cachelines are non-busy and the lemma is 
true.

We consider cases where the L1 cacheline is busy after the transition.  
Only {\it L1Miss} can make an L1 cacheline busy from non-busy and the 
rule requires a request to be waiting at the head of the \mrq buffer.  
If the L1 cacheline stays busy, then no rule can remove the request 
from the \mrq buffer. By the induction hypothesis, the lemma is true after 
any transition.
\end{proof}

\begin{lemma} \label{lemma:l2-busy-req}
If an L2 cacheline is busy, there must exist a request with the same 
address at the head of the \cprq buffer in L2.
\end{lemma}

\begin{proof}%[Lemma~\ref{lemma:l2-busy} Proof]
The proof follows the same structure as the previous proof for 
Lemma~\ref{lemma:l1-busy-mrq}.
\end{proof}

\begin{lemma} \label{lemma:getx}
For a memory request in a \cprq buffer, its \type and \pts equal
the \type and \pts of a pending processor request to the same address 
at the head of the \mrq at the L1 cache.
\begin{proof}
By Lemmas~\ref{lemma:req-mutex} and \ref{lemma:l1-busy-mrq}, the L1 
cacheline with the same address must be \busy and a pending processor 
request exists at the head of the \mrq buffer. Only the {\it L1Miss} 
rule sets the \type and \pts of a memory request in a \cprq buffer and 
they equal the \type and \pts from the processor request.
\end{proof}
\end{lemma}

\begin{lemma} \label{lemma:tox}
For a memory response in a \pc buffer, its \type equals the \type of a 
pending processor request to the same address at the L1 cache and if 
the \type = \Sh, its \rts is no less than the \pts of that processor 
request.
\begin{proof}
Similar to the proof of Lemma~\ref{lemma:getx}, the processor request 
must exist. Only the {\it ShReq\_S} and {\it ExReq\_S} rules set the 
\type and \rts of the response, and \type equals the \type of a memory 
request and if \type = \Sh, \rts is no less than the memory request.  
Then the lemma is true by Lemma~\ref{lemma:getx}.
\end{proof}
\end{lemma}

\begin{lemma} \label{lemma:hit-l2resp}
When the {\it L2Resp} rule fires, a request with the same address at 
the head of \mrq will transition from an L1 miss to an L1 hit.
\begin{proof}%[Lemma~\ref{lemma:hit-l2resp} Proof]
Before the transition of {\it L2Resp}, the L1 cacheline is busy, and a
response is at the head of the \pc buffer.  By Lemma~\ref{lemma:tox},
if the pending processor request has \type = \M, then the memory 
response also has \type = \M and thus it is an L1 hit. If the pending 
processor has \type = \Sh, also by Lemma~\ref{lemma:tox}, the memory 
response has \type = \Sh and the \rts of the response is no less than 
the \pts of the pending request. Therefore, it is also an L1 hit.
\end{proof}
\end{lemma}

\begin{lemma}[Coverage] \label{lemma:coverage} The union of all the 
entries in \cref{tab:lattice} is \true.
\end{lemma}

\begin{proof}%[Lemma~\ref{lemma:coverage} Proof]
By Lemma~\ref{lemma:l1-busy}, if {\it L1.busy} we can prove that {\it 
c2pRq\_exist} $\vee$ {\it p2cRp\_exist} $\Rightarrow$ \true.

Then it becomes obvious that the union of all the entries is true.
\end{proof}

\begin{lemma}[Mutually Exclusive] \label{lemma:mutex} The intersection 
of any two entries in \cref{tab:lattice} is \false.
\end{lemma}

\begin{proof}%[Lemma~\ref{lemma:mutex} Proof]
For most pairs of entries, we can trivially check that the 
intersection is \false. The only tricky cases are the intersection of 
entry 9 or 10 with an entry from 3 to 8. These cases can be proven 
\false using Lemma~\ref{lemma:req-mutex}, which implies that {\it 
c2pRq\_exist $\wedge$ p2cRp\_exist $\Rightarrow$ \false}.
\end{proof}

\begin{proof}[Lemma~\ref{lemma:lattice} Proof]
We need to prove two goals. First, for each transition rule except 
{\it Downgrade}, {\it WriteBackReq} and {\it WriteBackResp}, at least 
one request will dequeue or move down the lattice.  Second, for all 
transition rules no request will move
up the lattice. % We now prove these them for each transition rule.

We first prove that a transition with respect to address $a_1$ never 
moves a request with address $a_2$ ($\neq a_1$) up its lattice. The 
only possible way that the transition affects the request with $a_2$ 
%CHECK I changed this slightly
is by dequeuing from a buffer which may make a request with $a_2$ being
the head of the buffer and thus becomes ready. However, this can only move 
the request with $a_2$ down the lattice.

Also note that each processor can only serve one request per address 
at a time, because the \mrq is a FIFO. Therefore, for the second goal 
we only need to prove that requests with the same address in other 
processors do not move up the lattice. We prove both goals for each 
transition rule. 

For {\it LoadHit} and {\it StoreHit}, a request always dequeues from 
the \mrq and the lemma is satisfied.

For the {\it L1Miss} rule, a request must exist and be in entry $1$ in 
a table before the transition. And since \busy $=$ \true after the 
transition, it must move down the lattice to one of entries from $2$ 
to $10$. Since the L1 cacheline state does not change, no other 
requests in other processors move in their lattice.  

For the {\it L2Resp} rule, according to Lemma~\ref{lemma:hit-l2resp}, 
a request will move from \ch.\miss to \ch.\textit{hit}. In the 
lattice, this corresponds to moving from entry $10$ to entry $11$, 
which is a forward movement. For another request to the same address, 
the only entries that might be affected are entry $4$, $5$ and $6$.  
%CHECK changed this slightly, merged sentences
However, since \pc is a FIFO and the response is ready in the \pc 
buffer before the transition, no {\it WBRq} can be ready in this \pc 
buffer for other requests with the same address and thus they cannot 
be in entry $5$ or $6$. If another request is in entry $4$, the 
%CHECK changed this sentence slightly -- added request
transition removes the response from the \pc and this may make the 
{\it WBRq} ready in \pc and thus the request moves down the lattice. In all cases, 
no other requests move up the lattice. 

For the {\it ShReq\_S} or {\it ExReq\_S} rule to fire, there exists a 
request in the \cprq buffer which means the address must be busy in 
the corresponding L1 (Lemma~\ref{lemma:req-mutex}) and thus a request 
exists in its \mrq and misses the L1 (Lemma~\ref{lemma:l1-busy-mrq}).  
%CHECK changed slightly
This request, therefore, must be in entry $8$ in 
\cref{tab:lattice}.  The transition will dequeue the request and 
%CHECK changed, does the request move up or forward or are they synonymous? If so, say move up
enqueue a response to \pc and thus moves the request down to entry $9$ 
or $10$. For all the other requests with the same address, they cannot 
be ready in the \cprq buffer since the current request blocks them, 
and thus they are not in entry $3$ to $8$ in the lattice.  For the 
other entries, they can only possibly be affected by the transition if 
the current request is dequeued and one of them becomes ready. This, 
however, only moves the request down the lattice.  

The {\it Req\_M} rule can only fire if a request is ready in \cprq and 
the L2 is in the \M state. According to Lemma~\ref{lemma:req-mutex} and 
Lemma~\ref{lemma:l1-busy-mrq}, there exists a request in one \mrq that 
is in entry $3$ in a table. After the transition, this request will 
move to entry 4 or 5 or 6 and thus down the lattice.  For all the 
other requests, similar to the discussion of {\it ShReq\_S} and {\it 
ExReq\_S}, they either stay in the same entry or move down the 
lattice.

Finally, we talk about the {\it Downgrade}, {\it WriteBackReq} and 
{\it WriteBackResp} rules. The {\it Downgrade} rule can only fire when 
the L1 cacheline is non-busy, corresponding to entry $1$ and $11$ if 
the request is from the same L1 as the cacheline being downgraded.  
Entry $1$ cannot move up since it is the first entry. If a request is 
in entry $11$, since it is an L1 hit now, the {\it Downgrade} rule 
does not fire. For a request from a different L1, the {\it Downgrade} 
%CHECK
rule may affect entry $5$ and $6$. However, it can only move the 
request from entry $5$ to $6$ rather than the opposite direction.   

For the {\it WriteBackReq} rule, if the L1 cacheline is in the \Sh state, 
then nothing changes but a message is dequeued from the \pc buffer 
which can only move other requests down the lattice. If the L1 
cacheline has \M state, then if a request to the same address exists 
in the current L1, the request must be a hit and thus {\it 
WirteBackReq} cannot fire.  For requests from other L1s, they can only 
be affected if they are in entry 4. Then, the current transition can 
only move them down the lattice.  

For the {\it WriteBackResp} rule, the L2 cacheline moves from the \M to 
the \Sh state. All the other requests can only move down their lattice due 
to this transition.  

Finally, we prove that {\it Downgrade}, {\it WriteBackReq} and
{\it WriteBackResp} can only fire a finite number of times 
without other transitions being fired. Each time {\it Downgrade} is 
fired, an L1 cacheline's state goes down. Since there are only a finite 
number of L1 cachelines and a finite number of states, {\it Downgrade} 
can only be fired a finite number of times. Similarly, each {\it 
WriteBackReq} transition consumes a {\it WBRq} message which can only 
be replenished by the {\it Req\_M} rule. And each {\it WriteBackResp} 
transition consumes a {\it WBRp} which is replenished by {\it 
Downgrade} and {\it WriteBackReq} and thus only has finite count.  
\end{proof}

\begin{proof}[Theorem~\ref{theorem:livelock} Proof]
If there exists a pending request from any processor, by 
Lemma~\ref{lemma:lattice}, some pending request will eventually 
dequeue or move down the lattice which only has a finite number of 
states.  For a finite number of processors, since the \mrq is a FIFO, 
only a finite number of pending requests can exist.  Therefore, some 
pending request will eventually reach the end of the lattice and 
dequeue, proving the theorem.
\end{proof}

\section{Main Memory} \label{sec:mem}

For ease of discussion, we have assumed that all the data fit in the 
L2 cache, which is not realistic for some shared memory systems.  
A multicore processor, for example, has an offchip main memory which 
does not contain timestamps. For these systems, the components in 
\cref{tab:format-mem} and transition rules in \cref{tab:mem-rules} 
need to be added.  And for the initial system state, all the data 
should be in main memory with all L2 cachelines in \I state, and \mts 
= 0.

\begin{table}%{width=\columnwidth}
	\caption{ System Components required for main memory.}
    \begin{center}
	{ \footnotesize
		\begin{tabular}{|c | p{2.8in} | p{1.2in}|}
			\hline
			Component & Format & Message Types
		\\ \hline
			\memrq & \memrq.\entry = (\type, \addr, \data) & {\it 
		MemRq}
		\\ \hline
			\memrp & \memrp.\entry = (\addr, \data) & {\it MemRp}
		\\ \hline
			\mem   & \mem[\addr] = (\data) & -
		\\ \hline
			\red{\mts}   & - & -
		\\ \hline
			\end{tabular}
    }
    \end{center}
	\label{tab:format-mem}
	\vspace{-0.25in}
\end{table}

\begin{table}%{width=\columnwidth}
	\caption{ State Transition Rules for Main Memory. }
    \begin{center}
	{ \footnotesize
		\begin{tabular}{| p{3in} | p{2.4in}|}
		\hline
		Rules and Condition & Action
	\\ \hline
		\underline{\bf L2Miss} \newline
		\letbf (\id, \type, \addr, \red{\pts}) = \cprq.\getmsg 
		\newline
		{\bf condition:} \p.[\addr].\state = \I
	& MemRq.enq(\Sh, \addr, \_) \newline
		\busy $\coloneqq$ \true
	\\ \hline
		\underline{\bf MemResp} \newline
		{\bf let} (\addr, \data) = MemRp \newline
		{\bf let} (\state, \textit{l2data} \busy, \owner, \red{\wts},
		\red{\rts}) = \p.[\addr]
	& \state $\coloneqq$ \Sh \newline
		\textit{l2data} $\coloneqq$ \data \newline
		\busy $\coloneqq$ \false \newline
		\red{\wts $\coloneqq$ \mts} \newline
		\red{\rts $\coloneqq$ \mts}
	\\ \hline
		\underline{\bf L2Downgrade} \newline
		{\bf let} (\state, \data, \busy, \owner, \red{\wts},
		\red{\rts}) =
		\p.[\addr] \newline
		{\bf condition:} \state = \M $\wedge$ \busy = \false
	& \pc.enq(\owner, {\it Req}, \addr, \_, \_, \_, \_) \newline
		\busy $\coloneqq$ \true
	\\ \hline
		%\cellcolor[rgb]{.8,.9,.8}
		\underline{\bf L2Evict} \newline
		{\bf let} (\state, \data, \busy, \owner, \red{\wts}, 
		\red{\rts}) =
		\p.[\addr] \newline
		{\bf condition:} \state = \Sh
	& %\cellcolor[rgb]{.8,.9,.8}
		\memrq.enq(\M, \addr, \data) \newline
		\state $\coloneqq$ \I \newline
		\red{\mts $\coloneqq$ max(\rts, \mts)}
	\\ \hline
	\end{tabular}
	}
    \end{center}
	\label{tab:mem-rules}
	\vspace{-0.25in}
\end{table}

Most of the extra components and rules are handling main memory 
requests and responses and the \I state in the L2. However, \mts is a 
special timestamp added to represent the largest \rts of the 
cachelines stored in the main memory. The \mts guarantees that 
cachelines loaded from the main memory have proper timestamps and thus 
can be properly ordered. 

Due to limited space, we only prove that the Tardis protocol with main 
memory still obeys sequential consistency (SC). The system can also 
be shown to be deadlock- and livelock-free using proofs similar to 
\cref{sec:freedom}. For the SC proof, we only need to show that 
Lemma~\ref{lemma:basic}, \ref{lemma:clean-latest}  and 
\ref{lemma:ts-store} are true after the main memory is added.

In order to prove these lemmas, we need the following simple lemma. 

\begin{lemma} \label{lemma:l2-i}
If an L2 cacheline is in the \I state, no clean block exists for the 
address.
\begin{proof}
We prove by induction on the transition sequence. The hypothesis is 
true for the base case since no clean block exists.
%CHECK don't understand this sentence
If an L2 cacheline moves from \Sh to \I (through the L2Evict rule), 
the clean block (L2 cacheline in S state) is removed and no clean 
block exists for that address. By the transition rules, while the L2 line 
stays in the \I state, no clean block can be created. By the induction 
hypothesis, if an L2 cacheline is in the \I state after a transition, then 
no clean block can exist for that address.
\end{proof}
\end{lemma}

For Lemma~\ref{lemma:basic}, we only need to include 
Lemma~\ref{lemma:l2-i} in the original proof. For 
Lemmas~\ref{lemma:clean-latest} and \ref{lemma:ts-store}, we need to 
show the following properties of \mts. 

\begin{lemma} \label{lemma:mts}
If an L2 cacheline is in the \I state, then the following statements are 
true.
\begin{itemize}
\setlength\itemsep{0em}
	\item \mts is no less than the \rts of all the copies of the 
	block.
	\item No store has happened to the address at $ts$ such that $ts 
	> \mts$.
	\item The data value of the cacheline in main memory comes from a
	store $St$ which happened before the current physical time. And no
	other store $St'$ has happened such that $St.ts < St'.ts
	\leq \mts$.
\end{itemize}
\begin{proof}
All the statements can be easily proven by induction on the transition 
sequence. For \Sh $\rightarrow$ \I of an L2 cacheline, since the end 
\mts is no less than the \rts, by Lemmas~\ref{lemma:clean-latest} and 
\ref{lemma:ts-store}, all three statements are true after the 
transition.

Consider the other case where the L2 cacheline stays in \I state.  
Since no clean block exists (Lemma~\ref{lemma:l2-i}), the copies of 
the cacheline cannot change their timestamps and no store can happen.  
By the transition rules, the \mts never decreases after the transition. So 
the hypothesis must be true after the transition.
\end{proof}
\end{lemma}

To finish the original proof, we need to consider the final case where 
the L2 cacheline moves from \I to \Sh state ({\it MemResp} rule). In 
this case, both \wts and \rts are set to \mts. By 
Lemma~\ref{lemma:mts}, both Lemma~\ref{lemma:clean-latest} and 
Lemma~\ref{lemma:ts-store} are true after the current transition.  

\section{Related Work} \label{sec:related}

Snoopy~\cite{goodman1983} and directory~\cite{censier1978, tang1976} 
cache coherence are both popular coherence protocols and are widely 
adopted in multicore processors~\cite{chrysos2012, tilera}, 
multi-socket systems~\cite{ziakas2010, anderson2003} and distributed 
shared memory systems~\cite{keleher1994, li1989}. The Tardis 
protocol~\cite{yu2015} is a different yet as powerful coherence 
protocol. Tardis is conceptually simpler than a directory protocol and 
has excellent scalability. Some other timestamp based coherence protocols 
have also been proposed in the literature~\cite{nandy1994, elver2014, 
lis2011, singh2013} but none of them are as simple and high performant 
as Tardis. 

Both model checking and formal proofs are popular in proving the 
correctness of cache coherence protocols. Model 
checking based verification~\cite{murphi, tlc, tlc2,
gigamax,part,sym1,sym2,CMP,CMPHier,AllenEmerson1,SorinPV,comp,Mccomp,mcc,
Delzanno,Fractal} is a commonly used technique, but even with
several optimizations, it does not scale to automatically verify
real-world systems.
%able to exhaustively check all the system states and prove 
%correctness for a small scale system.
%
%However, it is not able to model 
%systems with realistic scale since computation grows exponentially 
%with the number of system states~\cite{??}.  

Many other works~\cite{plakal1998, afek1993, martin2003, brown1990, 
ladan2008} prove the correctness of a cache coherence protocol by 
proving invariants as we did in this paper. Our invariants are in 
general simpler than what they had partly because Tardis is simpler 
than a directory coherence protocol. Finally, our proofs can be 
machine-checked along the lines of the proofs for a hierarchical cache 
coherence protocol~\cite{muralidaran2015, park}.

\section{Conclusion} \label{sec:conclusion}

We provided simple, yet rigorous proofs of correctness for a 
recently-proposed scalable cache coherence protocol. Future work 
includes generalizing the protocol to relaxed memory consistency 
models and proving correctness, and machine-checking the proofs.

{
%	\footnotesize
%	\vspace{-0.1in}
%	\scriptsize
	\bibliographystyle{abbrv}
	\bibliography{refs,murali}
}

\end{document}